\documentclass[11pt,a4paper]{article}
\usepackage{a4wide}
\usepackage[margin=1in]{geometry}
\usepackage[utf8]{inputenc}
\usepackage{enumerate}
\usepackage{amsmath}
\usepackage{amsfonts}
\usepackage{nicefrac}
\usepackage{amssymb}
\usepackage{amsthm}
\usepackage{mathtools}
\usepackage{dsfont}
\usepackage{tikz}
\usetikzlibrary{calc}
\usepackage{pgfplots}
\pgfplotsset{compat=newest}
\usepackage{pgfplotstable}
\usepackage{xcolor}
\usepackage{thm-restate}

\pgfplotsset{
    discard if not/.style 2 args={
        x filter/.code={
            \edef\tempa{\thisrow{#1}}
            \edef\tempb{#2}
            \ifx\tempa\tempb
            \else
                
            \fi
        }
    }
}
\pgfmathdeclarefunction{lg2}{1}{%
    \pgfmathparse{ln(#1)/ln(2)}%
}
\pgfmathdeclarefunction{lg10}{1}{%
    \pgfmathparse{ln(#1)/ln(10)}%
}

\usepackage{tabularx}
\usepackage{booktabs}
\usepackage{paralist}

\usepackage{algorithm2e}
\DontPrintSemicolon

\usepackage[pagebackref,pdfdisplaydoctitle,menucolor=orange!40!black,filecolor=magenta!40!black,urlcolor=blue!40!black,linkcolor=red!40!black,citecolor=green!40!black,colorlinks]{hyperref}

\renewcommand*{\backref}[1]{}
\renewcommand*{\backrefalt}[4]{%
	\ifcase #1%
	\or [p.~#2.]%
	\else [pp.~#2.]%
	\fi%
}
\usepackage[square,numbers]{natbib}
\usepackage[nameinlink,sort&compress,capitalize]{cleveref}

\usepackage[backgroundcolor=gray!10,textsize=footnotesize]{todonotes}

\newtheorem{question}{Open Question}
\newtheorem{theorem}{Theorem}
\newtheorem{lemma}{Lemma}

\newtheorem{corollary}{Corollary}

\newtheorem{claim}{Claim}
\theoremstyle{definition}
\newtheorem{definition}{Definition}

\newtheorem{conjecture}{Conjecture}

\crefname{rrule}{Rule}{Rules}

\newcommand{\defproblemu}[3]{
  \vspace{2mm}
\noindent\fbox{
  \begin{minipage}{0.963\columnwidth}
  \textsc{#1} \\
  {\bf{Input:}} #2  \\
  {\bf{Question:}} #3
  \end{minipage}
  }
  \vspace{2mm}
}

\DeclarePairedDelimiterX{\abs}[1]{\lvert}{\rvert}{#1}
\DeclarePairedDelimiterX{\norm}[1]{\lVert}{\rVert}{#1}
\DeclarePairedDelimiterX{\ceil}[1]{\lceil}{\rceil}{#1}
\DeclarePairedDelimiterX{\angled}[1]{\langle}{\rangle}{#1}

\newcommand{\NN}{\mathbb{N}}

\newcommand{\yes}{\emph{yes}}
\newcommand{\no}{\emph{no}}

\newcommand{\eps}{\varepsilon}

\DeclareMathOperator{\tw}{tw}

\DeclareMathOperator{\pw}{pw}

\DeclareMathOperator{\cc}{cc}

\DeclareMathOperator{\Bin}{B_{\mathrm{in}}}
\DeclareMathOperator{\Bout}{B_{\mathrm{out}}}
\DeclareMathOperator{\Brin}{B_{\mathrm{rin}}}
\DeclareMathOperator{\Brout}{B_{\mathrm{rout}}}
\DeclareMathOperator{\outdeg}{\mathrm{outdeg}}
\DeclareMathOperator{\indeg}{\mathrm{indeg}}

\DeclareMathOperator{\Vinl}{V_{\mathrm{in}}^{0}}
\DeclareMathOperator{\Vinr}{V_{\mathrm{in}}^{1}}

\DeclareMathOperator{\Voutl}{V_{\mathrm{out}}^{0}}
\DeclareMathOperator{\Voutr}{V_{\mathrm{out}}^{1}}
\DeclareMathOperator{\bfI}{\mathbf{I}}
\DeclareMathOperator{\bfO}{\mathbf{O}}
\DeclareMathOperator{\bfVin}{\mathbf{V}_{\mathrm{in}}}
\DeclareMathOperator{\bfVout}{\mathbf{V}_{\mathrm{out}}}
\DeclareMathOperator{\sI}{s_{\bfI}}
\DeclareMathOperator{\sO}{s_{\bfO}}
\DeclareMathOperator{\sVin}{s_{\bfVin}}
\DeclareMathOperator{\sVout}{s_{\bfVout}}
\DeclareMathOperator{\symI}{\cdot\mathbf{I}\cdot}
\DeclareMathOperator{\symO}{\cdot\mathbf{O}\cdot}
\newcommand{\NULL}{\mathbf{NULL}}
\newcommand{\treedecomp}{\mathbb{T}}

\title{New Parameterized and Exact Exponential Time Algorithms for Strongly Connected Steiner Subgraph\thanks{AJA,JN and SW are supported by the project COALESCE that has received funding from the European Research Council (ERC), grant agreement No 853234.}} 

\author{
  Afrouz Jabal Ameli\thanks{
    Department of Information and Computing Sciences, Utrecht University, The Netherlands.\\
    Email addresses:
    \texttt{a.jabalameli@uu.nl}, \texttt{j.nederlof@uu.nl}, \texttt{s.wang5@uu.nl}
  }
  \and
  Tomohiro Koana\thanks{
    Department of Mathematical Informatics, The University of Tokyo, Japan.\\ Email addresses: \texttt{tomohiro.koana@gmail.com}
  }
  \and
  Jesper Nederlof\footnotemark[2]
  \and
  Shengzhe Wang\footnotemark[2]
}

\date{}

\begin{document}
\maketitle

\begin{abstract}
    The \emph{Strongly Connected Steiner Subgraph} (SCSS) problem is a well-studied network design problem that asks for a minimum subgraph that strongly connects a given set of terminals. In this paper, we present several new algorithmic and complexity results for SCSS.

    As our main result, we show that SCSS can be solved in time $17^{\tw} n^{O(1)}$ on directed graphs with $n$ vertices when a tree decomposition of the underlying graph of width $\tw$ is provided. This improves over a natural $\tw^{O(\tw)}n^{O(1)}$ time algorithm, and is the first algorithm with this kind of running time for a problem involving strong connectivity.

    Second, we give an exact exponential-time algorithm that solves SCSS in $2^n n^{O(1)}$ time, improving the known bounds for general directed graphs. 

    Finally, we investigate kernelization with respect to vertex cover. We prove that SCSS does not admit a polynomial kernel when parameterized by the size of a vertex cover, unless the polynomial hierarchy collapses. In contrast, we show that the closely related \emph{Strongly Connected Spanning Subgraph} problem \emph{does} admit a polynomial kernel under the same parameterization.
\end{abstract}

\section{Introduction}
\label{sec:intro}
The \textsc{Steiner Tree} problem is one of the earliest and most fundamental problems in combinatorial optimization: Given an undirected graph $G=(V,E)$ along with a set $T \subseteq V$ of terminals, the objective is to find a subgraph of $G$ with the minimum number of edges that connects all the terminals with each other. Accordingly, the problem is called \textsc{Steiner Tree} since this graph can be assumed to be a tree. The parameterized complexity of the problem has been studied already since the 1970s, when Dreyfus and Wagner showed that the problem can be solved in $O^*(3^{|T|})$ time,\footnote{We use the $O^*()$ notation to hide factors polynomial in the input size.} and hence proved it to be fixed parameter tractable when parameterized by the number of terminals. Parameterized by the number of vertices of $G$, there is an $O^*(1.36^n)$-time algorithm~\cite{FominGKLS13}.

Another important parameter for \textsc{Steiner Tree} is the \emph{treewidth} of the input graph $G$, which measures how tree-like $G$ is. A simple dynamic programming algorithm (cf e.g.~\cite[Section 7.3.3]{CyganFKLMPPS15}) shows that an instance of \textsc{Steiner Tree} can be solved in $O^*(\tw^{\tw})$ time when a tree decomposition of $G$ of width $\tw$ is given.\footnote{This is a standard assumption made when parameterized by the treewidth of the input graph, and hence will not be explicitly stated in this paper.}
The $\tw^{\tw}$ dependency here comes from the apparent need to store the \emph{connectivity pattern} of partial solutions. However, in~\cite{CyganNPPRW22} it was shown that storing all such patterns is not needed and that, somewhat surprisingly, this running time can be improved to $O^*(2^{O(\tw)})$ at the cost of randomization. This was later derandomized in a follow-up paper to the conference version of~\cite{CyganNPPRW22}, namely~\cite{BodlaenderCKN15}.

Directed and undirected generalizations of \textsc{Steiner Tree}, such as \textsc{Directed Steiner Network}, \textsc{Group Steiner Tree}, and \textsc{Steiner Forest}, have also been studied extensively under these parameters. Examples include \textsc{Steiner Forest} parameterized by treewidth or vertex cover~\cite{FeldmannL25}, \textsc{Group Steiner Tree} parameterized by treewidth~\cite{ChalermsookDLV17}, and \textsc{Directed Steiner Tree} parameterized by the treewidth of the solution network~\cite{FeldmannM23}.

The most direct directed analogue of \textsc{Steiner Tree} is the following problem:

\defproblemu{Strongly Connected Steiner Subgraph (SCSS)}{A directed graph $D=(V,A)$, set of terminals $T$, and integer $t$.}{A set $X \subseteq A$ such that $|X|\leq t$ and for each $u,v \in T$ there is a $uv$-path formed by arcs of $X$.}

We are interested in finding out which results that hold for \textsc{Steiner Tree} translate to the directed setting of \textsc{SCSS} (and its variants). Since a solution of \textsc{SCSS} need not be a tree, the problem is typically more complex than its undirected counterpart. Recall that, \textsc{Steiner Tree} is FPT, whereas~\textsc{SCSS} is W[1]-hard~\cite{GuoNS09} and hence probably not FPT. On the positive side, the problem can be solved in $n^{O(|T|)}$ time, i.e., it is XP parameterized by $|T|$.

The complexity of \textsc{SCSS} has been studied extensively from several perspectives. For example, in the setting of planar graphs, parameterized sub-exponential time algorithms are known for the problem~\cite{ChitnisFHM20, FeldmannM23}. From the approximation point of view, there is an $O(\log^{2} n)$-approximation algorithm for \textsc{SCSS}, and this is tight in the sense that it is hard to approximate with factor $O(\log^{2-\varepsilon} n)$~\cite{SteinerSCS2003}.

\subparagraph{Problems Related to \textsc{SCSS}.}
To place our results in context, we also describe three related problems. See Figure~\ref{fig:probs} for the relationships among them. First, setting $T=V$ in \textsc{SCSS}, we obtain the following problem:

\defproblemu{Strongly Connected Spanning Subgraph (SCSpS)}{A directed graph $D=(V,A)$, and an integer $t$.}{A set $X \subseteq A$ such that $|X|\leq t$ and $(V,X)$ is strongly connected.}

An instance of \textsc{SCSpS} is automatically a NO-instance if $D$ is not strongly connected. If we relax the requirement that $D$ itself be strongly connected, we obtain the following problem:

\defproblemu{Minimum Equivalent Graph (MEG)}{A directed graph $D=(V,E)$, and an integer $t$.}{A set $X \subseteq E$ such that $|X|\le t$ and for each $uv$-path in $D$, there is a $uv$-path in $(V,X)$.}

The problem is well-motivated, since one often wants to preserve all connectivity properties while working with a sparser network. We refer to~\cite[Chapter 12]{BangGutinbook} for more information on \textsc{SCSpS} and \textsc{MEG}. It is shown in~\cite[Section 2.3]{BangGutinbook} that a given instance $(D,t)$ of \textsc{MEG} can be solved in polynomial time provided that, for each strongly connected component $C_i$ of $D$, one is given a minimum-cardinality spanning subdigraph of $D[C_i]$. Hence an $O^*(c^n)$-time (respectively, $O^*(c^{\tw})$-time) algorithm for \textsc{SCSpS} implies an $O^*(c^n)$-time (respectively, $O^*(c^{\tw})$-time) algorithm for \textsc{MEG}.

In~\cite{FominLPS16} it was shown that \textsc{SCSpS} (and hence, \textsc{MEG}) can be solved on $n$-vertex graphs in time $O(2^{4 \omega n}m n)$, where $\omega$ denotes the matrix multiplication constant. A randomized $O^*(4^n)$-time algorithm is given in \cite{EibenKW25}.

We will show that our methods can also be used to address $2$-connectivity constraints in undirected graphs via the following problem:
 
\defproblemu{$2$-Edge Connected Spanning Subgraph ($2$-\textsc{ECSS})}{An undirected graph $G=(V,E)$ and an integer $t$.}{A set $X \subseteq E$ such that $|X|\leq t$ and $(V,X)$ is $2$-edge connected.}

Note that \textsc{$2$-ECSS} can be reduced to \textsc{SCSpS} by replacing each edge with an arc in both directions.\footnote{We spell this out in Appendix~\ref{sec:red}.} Note that \textsc{$2$-ECSS} generalizes the problem of finding a Hamiltonian cycle in an undirected graph. The $2$-\textsc{ECSS} problem has been a very active topic of research. Very recently, a line of work has focused on improving its approximability~\cite{2ECSSSTOC2025,kobayashi2ECSS,2ecssSoda2023}, culminating in a $(5/4 - \delta)$-approximation by~\cite{HommelsheimLL26}, for some constant $\delta\ge 10^{-6}$. 

\begin{figure}
\begin{center}
\begin{tikzpicture}[
    node distance=0.2cm and 2cm,
    every node/.style={
        draw,
        rectangle,
        rounded corners,
        minimum width=2.6cm,
        minimum height=0.9cm,
        align=center
    },
    arrow/.style={->, line width=0.6mm}
]

% Nodes (compact layout)
\node (n2) {\textsc{SCSpS}};
\node (n1) [left=of n2] {\textsc{SCSS}};
\node (n3) [above right=of n2] {\textsc{MEG}};
\node (n4) [below right=of n2] {\textsc{$2$-ECSS}};

% Arrows
\draw[arrow] (n1) -- (n2);

% Curved bidirectional arrows between n2 and n3
\draw[arrow, bend left=20] (n2) to (n3);
\draw[arrow, bend left=20] (n3) to (n2);

\draw[arrow] (n2) to (n4);

\end{tikzpicture}
\end{center}
\caption{The four problems \textsc{Strongly Connected Steiner Subgraph (SCSS)}, \textsc{Strongly Connected Spanning Subgraph (SCSpS)}, \textsc{Minimum Equivalent Graph (MEG)} and \textsc{$2$-Edge Connected Spanning Subgraph ($2$-ECSS)} studied in this paper. An arrow denotes the relation ``can be efficiently reduced to'', and hence an algorithm for \textsc{SCSS} implies an algorithm for all other problems with comparable run time.}
\label{fig:probs}
\end{figure}
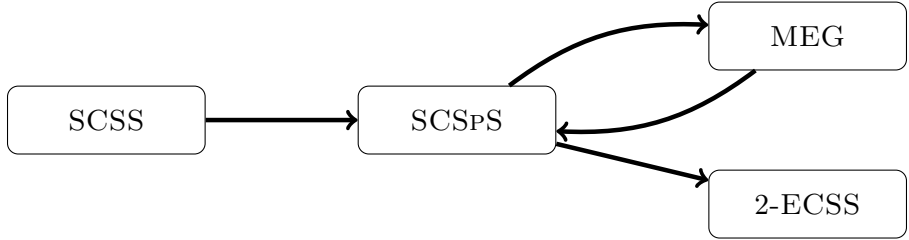

\subsection{Our results.}
\subparagraph{Treewidth Parameterization.}
We first study the complexity of \textsc{SCSS} parameterized by the treewidth of the underlying undirected graph. In light of the aforementioned surprisingly efficient method that deals with undirected connectivity~\cite{BodlaenderCKN15,CyganNPPRW22}, it is natural to wonder whether a similar approach can be used to deal with strong connectivity as well. We show that, at least for \textsc{SCSS}, this is indeed the case:
\begin{restatable}{theorem}{thmtwMain}\label{tw:main}
    There is a Monte Carlo algorithm that solves \textsc{SCSS} in $O^*(17^{\tw})$ time, where $\tw$ is the treewidth of the underlying undirected graph of $D$.
\end{restatable}
% \begin{theorem}\label{tw:main}
% There is a Monte Carlo algorithm that solves \textsc{SCSS} in $O^*(17^{\tw})$ time, where $\tw$ is the treewidth of the undirected graph obtained from $D$ by omitting directions.
% \end{theorem}
As with the method of~\cite{CyganNPPRW22}, our general idea can be applied to many problems and settings. We expect that it also works for other local problems augmented with a strong connectivity constraint.

To prove Theorem~\ref{tw:main} we use the characterization that a directed graph $D$ is strongly connected if and only if, for an arbitrarily chosen vertex $r$, there exist both an in-branching and an out-branching in $D$ that are rooted at $r$. This characterization has been used in previous related works as well~\cite{DBLP:journals/algorithmica/AgrawalMPS22,FominLPS16}, but these works were only able to get $2^{O(n)}$-time algorithms. To obtain our algorithm, we combine this characterization with the methods from~\cite{CyganNPPRW22} in a more efficient, but still highly non-trivial, way.

Using the mentioned reductions, we obtain the following direct consequence of Theorem~\ref{tw:main}:
\begin{corollary}\label{cor:main}
There is a Monte Carlo algorithm that solves \textsc{SCSpS}, \textsc{MEG} and \textsc{$2$-ECSS} in $O^*(17^{\tw})$ time, where $\tw$ is the treewidth of the underlying undirected graph.
\end{corollary}
In~\cite{BergerG07}, the authors solve a weighted variant of \textsc{$2$-ECSS} in $O^*(\tw^{O(\tw)})$ time.

\subparagraph{Exact Exponential Time Algorithm.}
\textsc{SCSS} generalizes the problem of detecting a Hamiltonian cycle in a directed graph. Improving exact exponential-time algorithms for Hamiltonian cycle is an active line of research; in particular, one seeks algorithms running in $1.999^n$ time on $n$-vertex graphs (see e.g.~\cite{BjorklundKK17}). Motivated by this line of work and the previously mentioned exact exponential-time algorithms for \textsc{Steiner Tree} and \textsc{MEG}, we study exact exponential-time algorithms for \textsc{SCSS}. Concretely, we prove the following:

\begin{theorem}\label{thm:scss-exact}
  \textsc{SCSS} can be solved in $O^*(2^n)$ time.
\end{theorem}

While Theorem~\ref{tw:main} used a characterization of strongly connected digraphs in terms of in/out-branchings, Theorem~\ref{thm:scss-exact} crucially relies on a characterization in terms of \emph{ear decompositions}. This is used in combination with a dynamic programming algorithm (over subsets of the vertex set) and the fast subset convolution method from~\cite{BjorklundHKK07}.

As before, we can use the aforementioned reductions to obtain the following corollary:
\begin{corollary}
    The \textsc{SCSpS}, \textsc{MEG} and \textsc{$2$-ECSS} problems on $n$-vertex graphs can be solved in $O^*(2^n)$ time.
\end{corollary}

Our result not only significantly improves the previously mentioned $O(2^{4 \omega n}m n)$-time algorithm for \textsc{MEG} from~\cite{FominLPS16}, but is also arguably much simpler.

\subparagraph{Kernelization parameterized by Vertex Cover.}
Finally, we investigate the complexity of these problems under parameterization by the vertex cover of the input graph. Parameterization by the vertex cover of a graph is a popular topic in parameterized complexity. For example, in~\cite{FeldmannL25} the authors studied the complexity of \textsc{Steiner Forest} parameterized by the vertex cover of the input graph. We study the \emph{kernelization complexity} of these problems, following earlier work on variants of the \textsc{Traveling Salesperson Problem}~\cite{BlazejCKSSV22,MannensNSS21}.
Concretely, we show:

\begin{theorem}\label{thm:ker}
\textsc{SCSpS} admits a quadratic kernel parameterized by the vertex cover of the underlying undirected graph, but \textsc{SCSS} does not admit a polynomial kernel parameterized by the vertex cover of the underlying undirected graph unless $\mathrm{NP} \subseteq \mathrm{coNP}/\mathrm{poly}$. 
\end{theorem}

The kernel is proved using matching arguments (via the \emph{expansion lemma} from~\cite{FLLSTZ19}) to exploit the fact that, for every vertex outside the independent set, both an incident out-arc and an incident in-arc must be selected. The conditional impossibility of a polynomial kernel is proven via a reduction to \textsc{Set Cover} parameterized by the universe size.

\subparagraph{Organization.}
In Section~\ref{sec:tw} we prove Theorem~\ref{tw:main}. In Section~\ref{sec:excexp} we prove Theorem~\ref{thm:scss-exact}, and in Section~\ref{sec:ker} we prove Theorem~\ref{thm:ker}. In Section~\ref{sec:con} we provide directions for further research, and in the Appendix~\ref{sec:red} we give the postponed reduction.

\section{Parameterization by Treewidth}\label{sec:tw}
In this section, we present the algorithm promised in Theorem~\ref{tw:main} for the \textsc{SCSS} problem parameterized by treewidth based on Cut\&Count technique, which was introduced in~\cite{CyganNPPRW11}.
\subsection{Preliminaries}

\paragraph*{Treewidth.}
\begin{definition}[Tree Decomposition~\cite{RobertsonS84}]\label{def:treeDecomp}
    A \emph{tree decomposition} of an (undirected or directed) graph~$G$ is a tree~$\treedecomp$ in which each vertex~$x \in \treedecomp$ has an assigned set of vertices~$\beta(x) \subseteq V$ (called a \emph{bag}) such that $\bigcup_{x \in \treedecomp} \beta(x) = V$ with the following properties:
    \begin{itemize}
        \item for any $(u, v) \in E$, there exists an~$x \in \treedecomp$ such that $u,v \in \beta(x)$.
        \item if $v \in \beta(x)$ and $v \in \beta(y)$, then $v \in \beta(z)$ for all $z$ on the path from $x$ to $y$ in $\treedecomp$.
    \end{itemize}
\end{definition}
The \emph{treewidth}~$\tw(\treedecomp)$ of a tree decomposition~$\treedecomp$ is the size of the largest bag of $\treedecomp$ minus one, and the treewidth of a graph $G$ is the minimum treewidth over all tree decompositions of~$G$.

Dynamic programming algorithms on tree decompositions are often presented on nice tree decompositions which were introduced by Kloks~\cite{Kloks94}.
We present a slightly refined definition of a nice tree decomposition of directed graphs from~\cite{CyganNPPRW22}.

\begin{definition}[Nice Tree Decomposition of Directed Graphs] \label{def:nicetreedecomp}
    A \emph{nice tree decomposition} of a directed graph is a tree decomposition with one special bag $z$ called the \emph{root} with $\beta(z) = \emptyset$ and in which each bag is one of the following types:
    \begin{description}
        \item[\upshape Leaf bag:]
            a leaf $x$ of $\treedecomp$ with $\beta(x) = \emptyset$.
        \item[\upshape Introduce vertex bag:]
            an internal vertex~$x$ of $\treedecomp$ with one child vertex~$y$
            for which $\beta(x) = \beta(y) \cup \{v\}$ for some
            $v \notin \beta(y)$. This bag is said to \emph{introduce} $v$.
        \item[\upshape Introduce arc bag:]
            an internal vertex~$x$ of $\treedecomp$ labeled with arc from $u$
            to $v$ with one child bag~$y$ for which
            $u,v \in \beta(x) = \beta(y)$. This bag is said to \emph{introduce}
            arc $(u,v)$.
        \item[\upshape Forget bag:]
            an internal vertex~$x$ of $\treedecomp$ with one child bag~$y$ for
            which $\beta(x) = \beta(y) \setminus \{v\}$ for some
            $v \in \beta(y)$. This bag is said to \emph{forget} $v$.
        \item[\upshape Join bag:]
            an internal vertex $x$ with two child vertices $l$ and $r$ with
            $\beta(x) = \beta(l) = \beta(r)$.
    \end{description}
    We additionally require that every arc is introduced exactly once. 
\end{definition}

Given a tree decomposition, a nice tree decomposition of equal width can be constructed in polynomial time~\cite{Kloks94, CyganNPPRW22}.
By fixing the root of $\treedecomp$, for two bags $x, y$ we say that $y$ is a descendant of $x$ if it is possible to reach $x$ when starting at $y$ and going only up the tree. 
We associate with each bag $x \in \treedecomp$ a vertex set $V_x \subseteq V$ where a vertex $v$ belongs to $V_x$ if and only if there is a bag $y$ which is a descendant of $x$ in $\treedecomp$ with $v \in \beta(y)$.
Note that $x$ is also a descendant of itself.
We also associate with each bag $x$ of $\treedecomp$ a subgraph of $G$ as follows:
\[ G_x = \Big{(}V_x, A_x = \{e |  \textrm{$e$ is introduced in a descendant of $x$ } \}\Big{)} \]
\paragraph*{Isolation Lemma.} 
\begin{definition}
    A function $\omega:U \rightarrow \mathbb{Z}$ \emph{isolates} a set family $\mathcal{F} \subseteq 2^U$ if there is a unique $S' \in \mathcal{F}$ with $\omega(S')= \min_{S \in \mathcal{F}}\omega(S)$. 
\end{definition}
This definition is useful because, in combination with the following lemma, it allows us to reduce the problem of deciding whether a set family is empty to computing the parity of its cardinality. Recall that for $X \subseteq U$, $\omega(X)$ denotes $\sum_{u \in X}\omega(u)$.

\begin{lemma}[Isolation Lemma~\cite{MulmuleyVV87}]\label{lem:iso}
    Let $\mathcal{F} \subseteq 2^U$ be a set family over a universe $U$ with $|\mathcal{F}|>0$.
    For each $u \in U$, choose a weight $\omega(u) \in \{1,2,\ldots,N\}$ uniformly and independently at random.
    Then $    	\Pr(\omega \textnormal{ isolates } \mathcal{F}) \geq 1 - \frac{|U|}{N}$.   
\end{lemma}

\paragraph*{Characterization via In-branching and Out-branching.}

\begin{definition}
    Given a directed graph $D = (V, A)$ and a set of terminal vertices $T$, 
    a \emph{Steiner in-branching} (respectively, \emph{out-branching}) is a weakly connected subgraph $B = (V_{B}, A_{B})$ of $D$ that satisfies the following:
    \begin{itemize}
        \item $T \subseteq V_{B}$.
        \item There exists a unique vertex $r \in V_{B}$, called root, with out-degree (respectively, in-degree) $0$.
        \item Every other vertex $w \in V_{B} \setminus \{r\}$ has out-degree (respectively, in-degree) exactly $1$.
    \end{itemize}
    We further call a subgraph $\widetilde{B}$ a \emph{relaxed Steiner in-branching} (respectively, \emph{out-branching}) if it satisfies the above degree conditions but is not necessarily weakly connected.
\end{definition}

Given a directed graph $D = (V, A)$ and terminals $T \subseteq V$, 
we fix an arbitrary vertex $r \in T$. 
\begin{definition}\label{def:compatibility}
A Steiner in-branching $\Bin$ and a Steiner out-branching $\Bout$ are \emph{compatible} if they span the same vertex set, i.e., $V(\Brin) = V(\Brout)$, 
and are both rooted at $r$.
\end{definition}
We note that a solution of \textsc{SCSS} can always be decomposed to a pair of compatible Steiner in-branching and out-branching.

\subsection{Cut \& Count}
We first provide a high-level overview of our approach.
In general, we assign random weights over edges such that we can isolate a minimum-weight solution with constant probability.
Then we define a superset of solutions that can be locally computed, and we augment each candidate solution with consistent cuts to distinguish true solutions.
The locality property makes it possible to count the number of solutions modulo 2 efficiently, where isolation further plays a role to decide the existence of a solution.

We start with the random weight assignment for edges. 
Note that a feasible  SCSS solution $X \subseteq A$ can be represented by a pair $(\Bin, \Bout)$ of compatible Steiner in-branching and out-branching. 
We thus call such a pair $(\Bin, \Bout)$ a \emph{solution} in the following.

We further take $U = A \times \{\symI, \symO\}$ where $\symI, \symO$ are symbols
and assign each element $e \in U$ a weight $\omega(e)$ chosen uniformly and independently at random from $\{1,2,\ldots, N\}$ where $N = 2|U|$.
In other words, for each arc $(u,v)$ we have two random weights $\omega(u,v,\symI)$ and $\omega(u,v,\symO)$ such that if we have a solution $(\Bin, \Bout)$, then the weight of the solution is 
\[\omega(\Bin, \Bout) = \sum_{(u,v) \in \Bin} \omega(u,v,\symI)+\sum_{(u,v) \in \Bout} \omega(u,v,\symO).\]
Then for an integer $W \in \{1,2,\ldots,2tN\}$, let $\mathcal{S}_{W}$ be the set of solutions of weight $W$.
If there does not exist a solution, it is clear that for any weight $W$ we have $\mathcal{S}_{W} = \emptyset$.
However, if there is a solution, then by Isolation Lemma~\ref{lem:iso}, for some weight $W\in \{1,2,\ldots,2tN\}$, with probability of at least $1/2$, we have $|\mathcal{S}_{W}| = 1$. In particular, the size of $|\mathcal{S}_{W}|$ is odd, and thus the decision problem can be reduced to the problem of counting the number of solutions of weight $W$ modulo 2.
We first introduce a set $\mathcal{C}_{W}$ such that $|\mathcal{C}_{W}| \equiv |\mathcal{S}_{W}| \pmod{2}$.
Unlike $\mathcal{S}_{W}$, the cardinality of $|\mathcal{C}_{W}|$ can be counted efficiently in $2^{O(\tw(D))}|A|^{O(1)}$ arithmetic operations, and we provide the formal proof of this later in Section~\ref{sec:countPart}.

\paragraph*{The Cut Part.}
To define $\mathcal{C}_{W}$, we start with defining a different set of \emph{candidate solutions} $\mathcal{R}_W$ that is a superset of solutions.
Those candidate solutions are local, meaning that they are easy to control using standard dynamic programming techniques on tree decompositions.

Namely, we relax the global connectivity constraint for candidate solutions in \textsc{SCSS}, and define $\mathcal{R}_{W}$ to be the family of pairs $(\Brin, \Brout)$ where $\Brin, \Brout$ are compatible and relaxed Steiner in/out-branchings with weight $\omega(\Brin, \Brout) = W$.
Note that we have fixed the root $r$ from the beginning of the problem.
In other words, candidate solutions still contain all terminal vertices but are not guaranteed to be strongly connected.
Due to the property of locality, the cardinality of $\mathcal{R}_{W}$ can be efficiently counted, while the parity of $|\mathcal{R}_{W}|$ does not necessarily match $|\mathcal{S}_{W}|$. However, by combining a cut of candidate solutions, we are ready to give the definition of $\mathcal{C}_{W}$.

\begin{definition}[Consistent Cut]
    For a directed graph $D = (V, A)$,  a cut $V^{0} \sqcup V^{1} = V$ of $D$ is \emph{consistent} if no arcs from $A$ are between $V^{0}$ and $V^{1}$.
\end{definition}

We define $\mathcal{C}_{W}$ to be the family of tuples of $(\Brin, V_{\mathrm{in}}^{0}, V_{\mathrm{in}}^{1})$ and $ (\Brout, V_{\mathrm{out}}^{0}, V_{\mathrm{out}}^{1})$ where $(\Brin, \Brout) \in \mathcal{R}_{W}$, and $(\Vinl, \Voutl)$ and $(\Voutl, \Voutr)$ are consistent cuts of the relaxed branchings $\Brin$ and $\Brout$ respectively. Formally, $\mathcal{C}_{W}$ is the following set:
\begin{equation*}
\left\{ 
    \begin{aligned} 
         (&\Brin, &&\Vinl, &&\Vinr, \\ 
          &\Brout, &&\Voutl, &&\Voutr) 
    \end{aligned} 
    \;\middle|\;
    \begin{array}{cl}
        \bullet & (\Brin, \Brout) \in \mathcal{R}_{W} \\[0.5em]
        \bullet & (\Vinl, \Vinr), (\Voutl, \Voutr) \text{ are consistent cuts respectively}\\[0.5em]
        \bullet & r \in \Vinl, \Voutl
    \end{array}
\right\}.
\end{equation*}
We further require root vertex $r \in \Vinl, \Voutl$ to break the symmetry as we shall explain later.

\paragraph*{The Count Part.}
We show that in $\mathcal{R}_{W}$ each solution of the problem is consistent with \emph{exactly one} cut, whereas all other candidate solutions are consistent with an \emph{even number} of cuts.

\begin{lemma}\label{lem:cutNum}
    Let $(\Brin, \Brout)$ be a candidate solution, the number of pairs of $(\Vinl, \Vinr)$ and $(\Voutl, \Voutr)$ such that $(\Vinl, \Vinr)$ is a consistent cut of $\Brin$ and $(\Voutl, \Voutr)$ is a consistent cut of $\Brout$ and $r \in \Vinl, \Voutl$ is $2^{(\cc(\Brin)-1)+(\cc(\Brout)-1)}$ where $\cc(D)$ denotes the number of connected components of $D$.
\end{lemma}
\begin{proof}
    For any connected component $C$ in $\Brin$, 
    because the cut is consistent,
    % by the definition of the consistent cut $(\Vinl, \Vinr)$, 
    either $C \subseteq \Vinl$ or $C \subseteq \Vinr$.
    Since we fix the choice for the component that contains the root $r$, and for the rest $\cc(\Brin) - 1$ components we can freely choose the side of the cut, we have $2^{\cc(\Brin) - 1}$ many consistent cuts.
    A similar analysis applies to $\Brout$, and in total this gives the lemma.
\end{proof}

We note that given a candidate solution $(\Brin, \Brout)$, if $\cc(\Brin) = \cc(\Brout) = 1$, then it is also a valid solution by the definition of relaxed Steiner branchings and compatibility.
Thus from Lemma~\ref{lem:cutNum}, it is clear that
\begin{equation}\label{eq:cutcount}
|\mathcal{C}_{W}| \equiv |\{(\Brin,\Brout) \in \mathcal{R}_{W} \mid \cc(\Brin) = \cc(\Brout) = 1\}| = |\mathcal{S}_{W}| \pmod{2}.
\end{equation}

Finally, we use dynamic programming on tree decomposition to count the cardinality of $\mathcal{C}_{W}$ modulo 2 efficiently.
\begin{lemma}\label{lem:cutCountAlgo}
    Given $D = (V, A)$, set of terminals $T$, and integer $t$, $\omega: A \times \{\symI, \symO\} \rightarrow \{1,2,\ldots,N\}$ and a nice tree decomposition $\mathbb{T}$ of width $\tw$, there exists an algorithm that can determine $|\mathcal{C}_{W}|$ modulo 2 for every $0 \le W \le 2tN$ in $17^{\tw}N^{2}|A|^{O(1)}$ time.
\end{lemma}

Equipped with Lemma~\ref{lem:cutCountAlgo} and~\eqref{eq:cutcount}, we can compute the parity of $|\mathcal{S}_W|$ for each $W=1,\ldots,2tN$ and combine with Lemma~\ref{lem:iso} to obtain Theorem~\ref{tw:main}.

\subsection{Proof of Lemma~\ref{lem:cutCountAlgo}}\label{sec:countPart}

To use dynamic programming, we first define ``partial solutions'' for the tree decomposition.
Recall that for a bag $x \in \treedecomp$,
we denote by $\beta(x) \subseteq V(D)$ the set of vertices that corresponds to $x$.
And we denote by $V_{x}$ the set of vertices in bags of all descendants of $x$, while by $D_{x}$ we denote the graph composed of vertices $V_{x}$ and the arcs $A_{x}$ introduced by the descendants of $x$.
For terminals, we let $T_{x} \subseteq V_{x}$ denote terminal vertices restricted to $V_{x}$.

For every bag $x \in \treedecomp$ of the tree decomposition, we consider partial solutions in the subgraph $D_{x}$.
\begin{definition}
    Given a bag $x \in \treedecomp$, let $D_{x} = (V_{x}, A_{x})$ with $\beta(x) \subseteq V_{x}$ and $T_{x} \subseteq V_{x}$.
    A subgraph $B \subseteq D_{x}$ is called a \emph{partial relaxed Steiner in(out)-branching} if it satisfies all conditions of a relaxed Steiner in(out)-branching, except for vertices in $\beta(x)$. Formally,
    \begin{itemize}
        \item $T_{x} \subseteq V(B)$.
        \item $\outdeg_{B}(r) = 0$ ($\indeg_{B}(r) = 0$), if $r \in V_{x}$.
        \item $\outdeg_{B}(u) \le 1$ ($\indeg_{B}(u) \le 1$), if $u \neq r \in \beta(x)$, o.w. $\outdeg_{B}(u) = 1$ ($\indeg_{B}(u) = 1$).
    \end{itemize}
\end{definition}

We note that the definition of compatibility naturally follows from Definition~\ref{def:compatibility} for partial relaxed Steiner in(out)-branchings.
Then for every bag $x \in \treedecomp$, integers $0 \le i \le |A|$, $0 \le w \le 2N|A|$, we define $\mathcal{R}_{x}(i, w)$ to be the following set:
\begin{equation*}
\left\{ (\Brin, \Brout) \;\middle|\;
    \begin{array}{cl}
        \bullet & \Brin, \Brout \text{ are partial relaxed Steiner in/out-branchings in } D_{x}\\[0.5em]
        \bullet &  \Brin, \Brout \text{ are compatible}\\[0.5em]
        \bullet & |E(\Brin \cup \Brout)| = i \\[0.5em]
        \bullet & \omega(\Brin, \Brout) = w
    \end{array}
\right\}.
\end{equation*}

The intuition for the set $\mathcal{R}_{x}(i,w)$ is that it contains all pairs of subgraphs of $D_{x}$ that could potentially be extended to a candidate solution in $\mathcal{R} = \bigcup_{W} \mathcal{R}_{W}$ with additional restrictions on cardinality and weight.
We define $\mathcal{C}_{x}(i, w)$ similarly by combining consistent cuts to be the following set:
\begin{equation*}
\left\{ 
    \begin{aligned} 
         (&\Brin, &&\Vinl, &&\Vinr, \\ 
          &\Brout, &&\Voutl, &&\Voutr) 
    \end{aligned} 
    \;\middle|\;
    \begin{array}{cl}
        \bullet & (\Brin, \Brout) \in \mathcal{R}_{x}(i, w) \\[0.5em]
        \bullet & (\Vinl, \Vinr), (\Voutl, \Voutr) \text{ are consistent cuts respectively}\\[0.5em]
        \bullet & r \in \Vinl, \Voutl, \text{ if }r \in V_{x}
    \end{array}
\right\}.
\end{equation*}
Further let $\mathcal{L} = \{\bfI, \bfO, \bfVin, \bfVout\}$ be the set of index labels, for vertex bag $\beta(x)$ we define a state as a function $s: \beta(x) \rightarrow \mathcal{X} = \{0, 1\}^{\mathcal{L}} \cup \{\NULL\}$ such that for any $u \in \beta(x)$ if $s(u) \neq \NULL$, then:
\[s(u) = (\sI(u), \sO(u), \sVin(u), \sVout(u)).\]

The various choices of $s$ actually describe all possible configurations of intersections of an element in $\mathcal{C} = \bigcup_{W}\mathcal{C}_{W}$, that is a valid solution with consistent cuts, with $\beta(x)$.
We thus further decompose each $\mathcal{C}_{x}(i,w)$ and define $\mathcal{A}_{x}(i, w, s)$ to be
\begin{equation*}
\left\{ 
    \begin{aligned} 
         (&\Brin, &&\Vinl, &&\Vinr, \\ 
          &\Brout, &&\Voutl, &&\Voutr) 
    \end{aligned} 
    \;\middle|\;
    \begin{array}{cl}
        \bullet & (\Brin, \Vinl, \Vinr, \Brout, \Voutl, \Voutr) \in \mathcal{C}_{x}(i, w) \\[0.5em]
        \bullet & s \text{ encodes vertex state}:\\[0.5em]
                & \quad \begin{aligned}
                    & \text{- If } u \notin V(\Brin) = V(\Brout), s(u) = \NULL\\
                    & \text{- o.w. }\\
                    & \quad \begin{aligned}
                        & \text{- } \indeg_{\Brout}(u) = \sI(u)\\
                        & \text{- } \outdeg_{\Brin}(u) = \sO(u)\\
                        & \text{- If } \sVin(u) = 0, u \in \Vinl, \text{ o.w. } u \in \Vinr\\
                        & \text{- If } \sVout(u) = 0, u \in \Voutl, \text{ o.w. } u \in \Voutr
                    \end{aligned}\\
                  \end{aligned} \\
                
    \end{array}
\right\}.
\end{equation*}
We let $DP_{x}(i, w, s) = |\mathcal{A}_{x}(i, w, s)|$ and note that to compute the value of $|\mathcal{C}_{W}|$ modulo 2, it suffices to compute the value of $DP_{r'}(t, W, \emptyset)$ for all $W$ where $t$ is the input integer and $r' \in \treedecomp$ is the root for the tree decomposition.

We now give the recurrence for $DP_{x}(i, w, s)$ for the dynamic programming algorithm.
For each bag $x \in \treedecomp$, we denote by $y, z$ the left and right child of $x$ if present.
All values of $DP_{x}(i, w, s)$ are set to 0 initially.
\paragraph*{Leaf Vertex Bag.}
\[DP_{x}(0, 0, \emptyset) = 1.\]

\paragraph*{Introduce Vertex Bag.}
Given $\beta(x) = \beta(y) \cup \{u\}$ where $u \in V$ is the introduced vertex, 
let $e$ be an element from the label set $\mathcal{X}$,
we define $s' = s[u \rightarrow e]$ as states update such that $s'(u) = e$ while $s'(v) = s(v)$ for any other $v \in \beta(y)$.
If $u$ is in partial solutions, then
\[\forall a, b \in \{0, 1\}, DP_{x}(i, w, s[u \rightarrow (0, 0, a, b)]) = DP_{y}(i, w, s)\]
If $u \notin T$, it is also possible to ignore $u$ in the partial solution, then
\[DP_{x}(i, w, s[u \rightarrow \NULL]) = DP_{y}(i,w,s).\]

\paragraph*{Forget Vertex Bag.}
Given $\beta(x) = \beta(y) \setminus \{u\}$ where $u$ is the forgotten vertex, let states update similarly defined as above.
If the vertex is a terminal but not a root, then it has to be covered by both branchings, giving the in-degree and out-degree both as 1,
\[DP_{x}(i, w, s) = \sum_{a, b} DP_{y}(i, w, s[u \rightarrow (1, 1, a, b)]).\]
Otherwise, if $u$ is indeed the root vertex $r$, then it is supposed to have degree 0 in both directions,
\[DP_{x}(i, w, s) = \sum_{a, b} DP_{y}(i, w, s[r \rightarrow (0, 0, a, b)]).\]
For $u \notin T$, we consider one extra case where $u$ is not included in partial solution in $D_{y}$,
\[DP_{x}(i, w, s) = \sum_{a, b} DP_{y}(i, w, s[u \rightarrow (1, 1, a, b)]) + DP_{y}(i, w, s[u \rightarrow \NULL]).\]

\paragraph*{Introduce Arc Bag.}
We consider the case where a single arc is introduced, and let $\beta(x) = \beta(y)$ and $E_{x} = E_{y} \cup (u, v)$.
Let $\mathbf{X} \in \mathcal{L}$ and $a \in \{0, 1\}$, we define another form of states update in a more subtle level as $s' = s[s_{\mathbf{X}}(u) \rightarrow a]$ such that $s'_{\mathbf{X}}(u) = a$ while all other values remain unchanged from $s$.
For the partial relaxed in-branching and out-branching, it is independent to decide whether to include the arc $(u,v)$. 
This yields four distinct cases for the transition, allowing us to express the recurrence as:
\[DP_{x}(i, w, s) = DP_{y}(i, w, s) + \Phi_{\Brin}^{uv} + \Phi_{\Brout}^{uv} + \Phi_{\Brin \wedge \Brout}^{uv}.\]
Here the base term $DP_{y}(i, w, s)$ captures the case where $(u,v)$ is used in neither branching.
The remaining terms $\Phi_{\Brin}^{uv}$, $\Phi_{\Brout}^{uv}$ and $\Phi_{\Brin \wedge \Brout}^{uv}$ represent the cases where $(u,v)$ is included only in the in-branching, only in the out-branching, or in both, respectively. We detail each of these terms below.
\begin{enumerate}
    \item To include the arc $(u, v)$ for $\Brin$ in $D_{x}$, the out-degree for vertex $u$ should be zero in $D_{y}$, and $u, v$ should stay on the same side of the cut. 
    Further we set the weight constraint to enforce that only $\Brin$ includes $(u, v)$.
    Thus we have
    \[\Phi_{\Brin}^{uv} = [s_{\bfO}(u) = 1 \wedge s_{\bfVin}(u) = s_{\bfVin}(v)] DP_{y}(k-1, w-\omega(u, v, \symI), s[s_{\bfO}(u) \rightarrow 0]).\]
    Here $[s_{\bfO}(u) = 1 \wedge s_{\bfVin}(u) = s_{\bfVin}(v)]$ is an indicator function where the value is 1 if all conditions are met, o.w. 0.
    We note that by having $s_{\bfVin}(u) = s_{\bfVin}(v)$, we implicitly exclude the possibility of $s(u)$ or $s(v)$ is $\NULL$.
    \item  The analysis for the case that only $\Brout$ includes $(u,v)$ is similar, and we have
    \[\Phi_{\Brout}^{uv} = [s_{\bfI}(v) = 1 \wedge s_{\bfVout}(u) = s_{\bfVout}(v)] DP_{y}(k-1, w-\omega(u, v, \symO), s[s_{\bfI}(v) \rightarrow 0]).\]
    \item If $(u, v)$ is utilized by both branchings, we combine above analysis to have
    \begin{align*}
    \Phi_{\Brin \wedge \Brout}^{uv} = & [s_{\bfO}(u) = 1 \wedge s_{\bfI}(v) = 1 \wedge s_{\bfVin}(u) = s_{\bfVin}(v) \wedge s_{\bfVout}(u) = s_{\bfVout}(v)] *\\
    & DP_{y}(k-1, w - \omega(u, v, \symI) - \omega(u, v, \symO), s[s_{\bfO}(u) \rightarrow 0, s_{\bfI}(v) \rightarrow 0])
\end{align*}
\end{enumerate}

We note that a special case in the recurrence is that one of the end points of the arc $(u, v)$ is the root vertex $r$ for branchings.
W.l.o.g., we assume $u = r$, then the arc $(r, v)$ can only be part of the out-branching. The recurrence is as follows:
\begin{align*}
    DP_{x}(i, w, s) & =  DP_{y}(i, w, s) + \Phi_{\Brout}^{rv}\\
    \Phi_{\Brout}^{rv} & = [s_{\bfI}(v) = 1 \wedge s_{\bfVout}(r) = s_{\bfVout}(v)] DP_{y}(k-1, w-\omega(r, v, \symO), s[s_{\bfI}(v) \rightarrow 0]).
\end{align*}
The analysis for the case where $v$ is the root vertex is similar.

\paragraph*{Join Bag.} 
Let $\beta(x) = \beta(y) = \beta(z)$ and $q = |\beta(x)|$. For states $s_1, s_2, s \in \mathcal{S}^{\beta(x)}$ where $\mathcal{S}$ is the set of all labels,
we say $s_1 \oplus s_2 = s$ if for each vertex $u \in \beta(x)$, the value of $s_1(u)$ and $s_2(u)$ can be combined into $s(u)$ under rules representing that partial solutions from two child bags can be unioned together, and we will explain the rules later in detail.
We have a general recurrence formula as:
\[DP_{x}(i,w,s) = \sum_{i_{1} + i_{2} = i} \sum_{w_{1} + w_{2} = w} \sum_{s_{1} \oplus s_{2} = s} DP_{y}(i_1, w_1, s_1)\cdot DP_{z}(i_2, w_{2}, s_{2}).\]
However, a straightforward computation of the above formula gives $289^{k}|A|^{O(1)}$ time complexity. 
We show that this bottleneck can be improved via fast subset convolution.

To apply subset convolution, we first decompose the state-combining rules.
We note that if $s_{1}$ and $s_{2}$ can be combined, for each vertex $u \in \beta(x)$, both states have to agree on the existence of $u$, i.e., 
\[s_{1}(u)=\NULL \Leftrightarrow s_{2}(u) = \NULL.\]
We introduce sub-state $c_{1}: \beta(x) \rightarrow \{0, 1\}$ to encode the existence of vertices, and for each fixed $c_{1}$ we only consider the state $s$ coincides with $c_{1}$. Namely, $c_{1}(u) = 0 \Leftrightarrow s(u) = \NULL$.
We focus on vertices involved and let $\beta'(x) := \{u \in \beta(x) \mid s(u) \neq \NULL\}$, and $q' = |\beta'(x)|$.
% and $s'$ be $s$ restricted to the set $\beta'(x)$.

We further require that states $s_{1}$ and $s_{2}$ agree on the cut information for any active vertex $u$ if they are to be combined.
Formally, this gives:
\[s^{1}_{\bfVin}(u) = s^{2}_{\bfVin}(u) \wedge s^{1}_{\bfVout}(u) = s^{2}_{\bfVout}(u).\]
We introduce another sub-state $c_{2}: \beta'(x) \rightarrow \{0,1\}^{\{\bfVin,\bfVout\}}$, and we are going to compute $DP_{x}(i, w, s)$ for all values such that $s$ agrees with $c_{1}$ and $c_{2}$,
i.e., for each vertex $u \in \beta(x)$, if $c_1(u) = 0$, then $s(u) = \NULL$, otherwise $s_{\bfVin}(u) = c^{2}_{\bfVin}(u)$ and $s_{\bfVout}(u) = c^{2}_{\bfVout}(u)$.
With this assumption, the combining table for $s$ is simplified since we only need to consider the value of $(s_{\bfI}(u), s_{\bfO}(u))$.
Table~\ref{tab:combineTable} shows how two individual states of a vertex $u$ in $y$ and $z$ combine to a state of $x$. XX indicates that two states do not combine. See e.g.~\cite{Rooij20} for more information about convolution tables.

\begin{table}[h]
    \centering
    \begin{tabular}{cccccc}
          $(s_{\bfI}(u), s_{\bfO}(u))$ &(0,0) &(0,1)&(1,0) &(1,1) \\
         (0,0)&  (0,0) & (0,1)  &(1,0)  & (1,1) \\
         (0,1)&  (0,1) &  XX    & (1,1) &  XX  \\
         (1,0)&  (1,0) & (1,1)  & XX    &   XX  \\
         (1,1)&  (1,1) & XX     & XX    & XX    \\
    \end{tabular}
    \caption{Degree state convolution table.}
    \label{tab:combineTable}
\end{table}
The structure of this table shows that combinations are only allowed when the non-zero degree requirements are disjoint.
It indicates that for fixed $i_{1} + {i_2} = i$, $w_{1} + w_{2} = w$ and sub-state $c_{1}$,$c_{2}$, computing $\sum_{s_{1} \oplus s_{2} = s} DP_{y}(i_{1}, w_{1}, s_{1}) \cdot DP_{z}(i_{2}, w_{2}, s_{2})$ where $s$ agrees with $c_{1}$ and $c_{2}$ is equivalent to performing a subset convolution over the set universe $S' = \beta'(x) \times \{\symI, \symO\}$.

To see that, consider a bijection between the valid states $s$ and the power set of $S'$.
For any state $s$, let $\chi_s \subseteq S'$ be the subset defined as
\[\chi_s = \{ (u, \symI) \mid u\in \beta'(x) \wedge s_{\bfI}(u) = 1 \} \cup \{ (u, \symO) \mid u\in \beta'(x) \wedge s_{\bfO}(u) = 1 \}.\]
From Table~\ref{tab:combineTable}, it can be verified that two states $s_{1}$ and $s_{2}$ can be validly combined if and only if $\chi_{s_{1}} \cap \chi_{s_{1}} = \emptyset$ and their union forms combined states as $\chi_{s_{1}} \cup \chi_{s_{1}} = \chi_{s_{1} \oplus s_{2}}$.
Since valid combinations correspond exactly to disjoint unions of subsets in $S'$, we can efficiently compute the summation using fast subset convolution.

Recall that $q = |\beta(x)|$ and $q' = |\beta'(x)|$,
using fast subset convolution, this can be computed in time $4^{q'}q'^{O(1)}$.
Further for a fixed $c_{1}$, there are $4^{q'}$ choices for $c_{2}$, and there are $\binom{q}{q'}$ ways to choose the $q'$ active vertices for $c_{1}$.
Considering $i \le t \le |A|$ and $w \le 2tN$, the total running time is upper bounded by,
\[\left(\sum_{0 \le q' \le q} \binom{q}{q'} 16^{q'}q'^{O(1)}\right)\cdot N^2|A|^{O(1)} \le 17^{\tw}N^2|A|^{O(1)}=O^*(17^{\tw}) .\]

\section{Exact Exponential Time Algorithms}\label{sec:excexp}
In this section, we describe our exact exponential time algorithm for SCSS based on a dynamic programming approach.
\subsection{Preliminaries}

We use the following algebraic tool in our computation, which we proceed to define.

\begin{definition}[(min,+) subset convolution]
For set functions $f,g : 2^V \to \NN \cup \{+\infty\}$, the (min,+) subset convolution $h=f * g$ is the set function $h:2^V \to \NN \cup \{+\infty\}$ defined for every $S\subseteq V$ by 
\[
  (f * g)[S] \;:=\; \min_{S = A \sqcup B} f[A] + g[B].
\]
\end{definition}
The (min,+) subset convolution can be computed efficiently (see Theorem~10.17 in~\cite{CyganFKLMPPS15}); We restate this result here for completeness.
\begin{lemma}[Fast evaluation of (min,+) subset convolution]
Let $n = |V|$ and $M \in n^{O(1)}$.
Given $f,g:2^V\to \{ 0, \dots, M \}$, the values $(f * g)[S]$ for all $S\subseteq V$ can be computed in $O^*(2^{n})$ time and $O^*(2^{n})$ space.
\end{lemma}

We also use the following characterization of strongly connected graphs,

\begin{definition}[Directed ear]
Let $H$ be a subdigraph of a digraph $D$.
A \emph{directed ear} $P$ with respect to $H$ is either a directed path in $D$ whose endpoints lie in $V(H)$ and whose internal vertices are outside $H$, or a directed cycle in $D$ that intersects $H$ in exactly one vertex.
Its \emph{internal vertex set} is $V^{\circ}(P):=V(P)\setminus V(H)$.
\end{definition}

\begin{definition}[Ear decomposition]
An \emph{ear decomposition} of a digraph $D$ is a sequence
\[
  \mathcal{P}=(P_1,P_2,\dots,P_r)
\]
of subdigraphs where $P_1$ is a directed cycle, and for each $i\ge 2$, $P_i$ is a directed ear with respect to $H_{i-1}:=P_1\cup\cdots\cup P_{i-1}$.
\end{definition}

\begin{lemma}[Characterization of strong connectivity, Lemma 5.3.2 in~\cite{BangGutinbook}]\label{lem:EDinStronglyConnected}
A digraph $D$ with at least two vertices is strongly connected if and only if it admits an ear decomposition.
\end{lemma}

\subsection{\texorpdfstring{$O^*(2^n)$}{O-star-2-to-n}-time algorithm}

In this section we prove Theorem~\ref{thm:scss-exact}. Let $(D=(V,A),T,t)$ be an instance of SCSS, and let $n:=|V|$.
If $|T|\le 1$, then the instance is trivially a \yes-instance if and only if $t \ge 0$, so assume $|T|\ge 2$.
By Lemma~\ref{lem:EDinStronglyConnected}, every strongly connected subgraph admits an ear decomposition, and conversely every ear decomposition yields a strongly connected subgraph.

  \subparagraph*{Precomputation.}
  For every ordered pair $(s,t)\in V\times V$, define a table
  \[
    F_{s,t}\colon 2^V \to \NN \cup \{+\infty\}
  \]
  by setting, for every $X\subseteq V$,
  \[
    F_{s,t}[X] :=
    \begin{cases}
      |X|+1, & \text{if $s,t\notin X$ and there exists a simple directed $(s,t)$-path $P$ in $D$}\\
             & \text{such that $V^{\circ}(P)=X$,}\\
      +\infty, & \text{otherwise.}
    \end{cases}
  \]
  If $s=t$, we interpret such a path as a directed cycle through $s$ whose internal vertex set is $X$.
  By a standard subset dynamic program for directed paths with prescribed vertex set, all tables $F_{s,t}$ can be computed in $O^*(2^n)$ time and space.

  \subparagraph*{Dynamic Programming.}
  For every $q\in \{1,\dots,n-1\}$, we define a table
  $T_q\colon 2^V \to \NN \cup \{+\infty\}$
  where
  \[
    T_q[X] :=
    \min\left\{\,
      |F|
      \;\middle|\;
      \begin{array}{l}
        F\subseteq A(D[X]) \text{ and } (X,F) \text{ admits an ear decomposition} \\
        \text{with exactly $q$ ears}
      \end{array}
    \right\}.
  \]
  If no such subgraph exists, we set $T_q[X]:=+\infty$.
  
  The initialization corresponds to a single directed cycle:
  \[
    T_1[X] = \min_{a\in X} F_{a,a}[X\setminus\{a\}]
    \qquad \text{for every } X\subseteq V.
  \]
  For every $q\in\{2,\dots,n-1\}$ and every ordered pair $(s,t)\in V\times V$, define the masked table
  \[
    S_{q-1,s,t}[X] :=
    \begin{cases}
      T_{q-1}[X], & \text{if } \{s,t\}\subseteq X,\\
      +\infty, & \text{otherwise.}
    \end{cases}
  \]
  We then compute $U_{q,s,t} := S_{q-1,s,t} * F_{s,t}$, where $*$ denotes the $(\min,+)$ subset convolution. Equivalently,
  \[
    U_{q,s,t}[X]
    =
    \min_{X=A\sqcup B} \bigl(S_{q-1,s,t}[A] + F_{s,t}[B]\bigr).
  \]
  Finally, we set
  \[
    T_q[X] = \min_{(s,t)\in V\times V} U_{q,s,t}[X]
    \qquad \text{for every } X\subseteq V.
  \]
  The optimum value of the \textsc{SCSS} instance is
  \[
    \min_{q\in\{1,\dots,n-1\}} \ \min_{T\subseteq X\subseteq V} T_q[X].
  \]

  \subparagraph*{Correctness.}
  We prove by induction on $q$ that $T_q[X]$ equals the minimum number of arcs of a subgraph of $D[X]$ that admits an ear decomposition with exactly $q$ ears.
  For $q=1$, this is exactly the initialization, since a one-ear decomposition is precisely a directed cycle.

  For $q\ge 2$, consider an ear decomposition with exactly $q$ ears on vertex set $X$, and let its last ear be a directed ear from $s$ to $t$ with internal vertex set $B$.
  Let $A:=X\setminus B$.
  Then the first $q-1$ ears form an ear decomposition on $A$, and necessarily $\{s,t\}\subseteq A$.
  Hence the cost of the first $q-1$ ears is captured by $S_{q-1,s,t}[A]$, while the last ear is captured by $F_{s,t}[B]$.
  Therefore the total cost appears in $U_{q,s,t}[X]$.

  Conversely, if $X=A\sqcup B$, $S_{q-1,s,t}[A]<+\infty$, and $F_{s,t}[B]<+\infty$, then the subgraph witnessing $S_{q-1,s,t}[A]$ admits an ear decomposition with $q-1$ ears and contains both $s$ and $t$, and the path witnessing $F_{s,t}[B]$ can be appended as the last ear.
  This yields an ear decomposition with exactly $q$ ears on $X$.
  Taking the minimum over all ordered pairs $(s,t)$ proves the recurrence.%\jn{Use lemma 3.6}

  \subparagraph*{Running Time.}
  There are $n^2$ tables $F_{s,t}$, and all of them can be computed in $O^*(2^n)$ time and space.
  For each $q\in\{2,\dots,n-1\}$ and each ordered pair $(s,t)$, one $(\min,+)$ subset convolution over universe $V$ is computed in $O^*(2^n)$ time and space.
  Since the number of values of $q$ and ordered pairs $(s,t)$ is polynomial in $n$, the overall running time is $O^*(2^n)$. This concludes the proof of Theorem~\ref{thm:scss-exact}.

\section{Parameterization by Vertex Cover}\label{sec:ker}
In this section we study kernelization algorithms for these problems parameterized by the size of vertex cover, say $k$. In particular, in Section~\ref{sec:SCSSKernel} we show 
that \textsc{SCSS} (and consequently $2$-ECSS) admits a $O(k^2)$ kernel. Then in Section~\ref{sec:SteinerSCSKernel} we prove that \textsc{SCSS} has no polynomial compression.
The obtained Theorems~\ref{thm:KernelizationSTeinerCSS} and~\ref{thm:NoKernel} together imply Theorem~\ref{thm:ker}.

\subsection{Kernel for Strongly Connected Spanning Subgraph}\label{sec:SCSSKernel}

We use the following reformulation of the $\ell=1$ case of the expansion lemma.
\begin{lemma}[Lemma 3.2 of~\cite{FLLSTZ19}]\label{lem:scss-new-expansion}
    Let $G$ be a bipartite graph with bipartition $(A,B)$. Then there exist sets $A_1\subseteq A$ and $B_1\subseteq B$ such that:
    \begin{itemize}
        \item there is a matching $M$ saturating $A_1$ into $B_1$,
        \item $N_G(B_1)\subseteq A_1$, and
        \item $|B\setminus B_1|\le |A\setminus A_1|$.
    \end{itemize}
    Moreover, such sets $A_1$ and $B_1$ can be found in polynomial time.
\end{lemma}

\begin{theorem}\label{thm:KernelizationSTeinerCSS}
    SCSS parameterized by vertex cover admits a kernel of $O(k^2)$ vertices, where $k$ is the size of the minimum vertex cover.
\end{theorem}
\begin{proof}
    Consider an instance $(D=(V,A),t)$ of SCSS, and let $S$ be a vertex cover of the underlying graph of $D$ of size at most $k$. Set $I:=V\setminus S$. Since $S$ is a vertex cover, $I$ is independent. Let $T:=S\times S$, and define a bipartite graph $B$ with bipartition $(T,I)$ by putting an edge between $(u,w)\in T$ and $v\in I$ if and only if $(u,v)$ and $(v,w)$ are arcs of $D$.

    By Lemma~\ref{lem:scss-new-expansion}, there exist sets $T_1\subseteq T$ and $I_1\subseteq I$ satisfying the three properties listed above.
    We use the following reduction rule.

    \subparagraph*{Reduction Rule.}
    Let $R:=I_1\setminus V(M)$. Replace $(D,t)$ by $(D-R,t-2|R|)$.
    \begin{claim}
        $(D,t)$ is a \yes-instance if and only if $(D-R,t-2|R|)$ is a \yes-instance.
    \end{claim}
    \begin{proof}
        Let $(D',t')$ be the reduced instance, where $D':=D-R$ and $t':=t-2|R|$.

        Assume first that $(D',t')$ is a \yes-instance, and let $H'$ be a strongly connected spanning subgraph of $D'$ with at most $t'$ arcs. If some vertex $v\in I$ has no neighbor in $B$, then $(D,t)$ is trivially a \no-instance, since in any strongly connected spanning subgraph the vertex $v$ must have an in-neighbor and an out-neighbor, both in $S$ because $I$ is independent. Hence we may assume that every vertex of $I$, and in particular every vertex of $R$, has a neighbor in $B$. For every $v\in R$, choose an arbitrary neighbor $(u,w)\in N_B(v)$ and add the arcs $(u,v)$ and $(v,w)$ to $H'$. Let $H$ be the resulting spanning subgraph of $D$. Since we add exactly two arcs for each vertex of $R$, we obtain
        \(
            |A(H)|=|A(H')|+2|R|\le t.
        \)
        Moreover, $H$ is strongly connected. Indeed, $H'$ is strongly connected on $V(D')$, and for every $v\in R$ the vertices $u,w$ belong to $S\subseteq V(D')$. Hence every old vertex reaches $u$ in $H'$ and therefore reaches $v$ in $H$, while $v$ reaches $w$ and therefore reaches every old vertex. It follows that every vertex of $D$ can reach $v$ and can be reached from $v$, so $H$ is strongly connected.

        Conversely, assume that $(D,t)$ is a \yes-instance, and let $H$ be a strongly connected spanning subgraph of $D$ with at most $t$ arcs. For every $\tau=(u,w)\in T_1$, let $m_\tau$ be the unique vertex of $I_1$ matched to $\tau$ by $M$. We construct a subgraph $H^*$ of $D'$ as follows. We delete all vertices of $I_1$ from $H$ and arcs incident to them,
        % keep every arc of $H$ whose endpoints both lie outside $I_1$, 
        add back exactly the matched vertices $m_\tau$ for $\tau\in T_1$, and for each $\tau=(u,w)\in T_1$ add the arcs $(u,m_\tau)$ and $(m_\tau,w)$.

        We first bound the number of arcs of $H^*$. Since $I$ is independent, $H$ has no arc with both endpoints in $I$. As $H$ is strongly connected, every vertex of $I_1$ has indegree at least one and outdegree at least one in $H$. Therefore $H$ contains at least $2|I_1|$ arcs incident with $I_1$. On the other hand, $H^*$ contains exactly $2|T_1|$ arcs incident with the vertices $m_\tau$, one entering and one leaving each matched vertex. Hence
        \(
            |A(H^*)|
            \le |A(H)|-2|I_1|+2|T_1|
            \le t-2|R|.
        \)

        It remains to prove that $H^*$ is strongly connected. Let
        \(
            U:=S\cup (I\setminus I_1).
        \)
        Consider any directed path $P$ in $H$ whose endpoints lie in $U$. Whenever $P$ contains a segment $a\to z\to b$ with $z\in I_1$, we have $a,b\in S$ because $I$ is independent. Hence $(a,b)\in N_B(z)\subseteq T_1$, and therefore the matched vertex $m_{(a,b)}$ is defined and the arcs $(a,m_{(a,b)})$ and $(m_{(a,b)},b)$ belong to $H^*$. Replacing every such segment $a\to z\to b$ by $a\to m_{(a,b)}\to b$ turns $P$ into a directed walk in $H^*$ with the same endpoints. Since $H$ is strongly connected, every two vertices of $U$ are joined by directed paths in both directions in $H$, and thus every two vertices of $U$ are mutually reachable in $H^*$.

        Finally, let $\tau=(u,w)\in T_1$. The vertex $m_\tau$ is incident in $H^*$ with the arcs $(u,m_\tau)$ and $(m_\tau,w)$, where $u,w\in S\subseteq U$. Since the vertices of $U$ are mutually reachable in $H^*$, 
        $m_\tau$ can reach and be reached by every vertex of $U$ (via $u$ and $w$, respectively).
        % every vertex of $U$ reaches $m_\tau$ via $u$, and $m_\tau$ reaches every vertex of $U$ via $w$. 
        Therefore every matched $m_\tau$ shares the same strongly connected component as $U$, and hence $H^*$ is strongly connected.
    \end{proof}

    To obtain the kernel of $O(k^2)$ vertices we do the following: if $|I|>|T|$, then we apply the reduction rule once. Let $(\widehat{D},\widehat{t})$ be the resulting instance, and let $\widehat{I}:=I\setminus R$. Since $M$ saturates $T_1$, we have $|V(M)|=|T_1|$, and by property (iii),
    \[
        |\widehat{I}|=|I\setminus I_1|+|V(M)|\le |T\setminus T_1|+|T_1|=|T|.
    \]
    Thus, after one application of the reduction rule, we may assume that $|I|\le |T|$. Since $T=S\times S$, it follows that
    \[
        |V(D)|=|S|+|I|\le |S|+|T|=|S|+|S|^2\le k+k^2=O(k^2),
    \]
    which proves the theorem.
\end{proof}

\subsection{No Kernel for SCSS}\label{sec:SteinerSCSKernel}
Our goal is to show that \textsc{SCSS}, parameterized by the size of a vertex cover, does not admit a polynomial kernel. In fact, we prove a stronger statement: it does not even admit a polynomial compression.

\begin{definition}[Polynomial compression, Definition 1.5 of~\cite{FLSkernelizationGreen2019}]
    Let $Q \subseteq \Sigma^* \times \mathbb{N}$ be a parameterized problem. A \emph{polynomial compression} for $Q$ is a polynomial-time algorithm that, given an instance $(x,k)$, outputs a string $y \in \Sigma^*$ such that $(x,k) \in Q \iff y \in L$ for some language $L \subseteq \Sigma^*$, and $|y| \leq p(k)$, for some polynomial $p$.
\end{definition}

In other words, a polynomial compression reduces an instance of a parameterized problem to an equivalent instance of \emph{some} problem whose size is bounded polynomially in the parameter. Since a polynomial kernel is a special case of a polynomial compression, ruling out polynomial compressions also rules out polynomial kernels.

To transfer compression lower bounds, we use polynomial parameter transformations.

\begin{definition}[Polynomial parameter transformation (PPT),  Definition 19.1 of~\cite{FLSkernelizationGreen2019}]
    Let $Q, Q' \subseteq \Sigma^* \times \mathbb{N}$ be parameterized problems. A \emph{polynomial parameter transformation} from $Q$ to $Q'$ is a polynomial-time algorithm that maps any instance $(x,k)$ of $Q$ to an instance $(x',k')$ of $Q'$ such that
    \[
    (x,k) \in Q \iff (x',k') \in Q' \text{ and } k' \leq p(k), \text{ for some polynomial $p$.}
    \]
\end{definition}

The importance of PPT reductions lies in their ability to propagate lower bounds.

\begin{lemma}[Theorem 19.2 of~\cite{FLSkernelizationGreen2019}]
    Let $Q$ and $Q'$ be parameterized problems such that $Q \leq_{\mathrm{ppt}} Q'$. If $Q$ does not admit a polynomial compression unless $\mathrm{NP} \subseteq \mathrm{coNP}/\mathrm{poly}$, then neither does $Q'$.
\end{lemma}

Therefore, to prove that \textsc{SCSS} parameterized by vertex cover does not admit a polynomial kernel, we show that there is a PPT reduction from \textsc{Set Cover} parameterized by the universe size. In order to prove this, we make use  of the fact that \textsc{Set Cover} parameterized by the universe size does not have a polynomial compression~(See \cite{LPRS17} or Theorem 19.5 in \cite{FLSkernelizationGreen2019}).

\begin{theorem}\label{thm:NoKernel}
    \textsc{SCSS} parameterized by vertex cover does not admit a polynomial compression unless $\mathrm{NP} \subseteq \mathrm{coNP}/\mathrm{poly}$.
\end{theorem}
\begin{proof}
    We provide a PPT transformation from \textsc{Set Cover} parameterized by the universe size to \textsc{SCSS} parameterized by vertex cover.
    Consider an instance $\mathcal I$ of \textsc{Set Cover} on universe $U=\{1,2,\ldots,n\}$, and family of sets $F=\{S_1,\ldots,S_m\}$. We create an instance $\mathcal I'$ of \textsc{SCSS} on graph $D$ with $m+n+2$ vertices, that has a vertex cover of size $n+1$. For this purpose, we define $V(D)=\{s,t\}\cup V_U\cup V_F$, where $V_U=\{u_1,\ldots,u_n\}$ (one vertex for each element) and $V_F=\{v_1,\ldots,v_m\}$ (one vertex for each of the sets). Now, $A(D)$ the set of arcs of $D$ are defined as follows:
    \begin{itemize}
        \item for every $i\in [m]$, $(s,v_i) \in A(D)$,
        \item for every $i \in [m]$ and $j \in [n]$ such that $j\in S_i$, $(v_i,u_j) \in A(D)$,
        \item for every $j\in [n]$, $(u_j,t)\in A(D)$,
        \item and $(t,s)\in A(D)$.
    \end{itemize}

    Finally, we define the set of terminals as $V_U$ and all the other vertices are Steiner nodes.
    Note that the size of the vertex cover of $D$ is polynomial in $n$, as $\{s\} \cup V_U$ is a vertex cover of size $n+1$ for $D$. 
    
     We now show that, given a feasible solution $F'\subseteq F$ for $\mathcal{I}$ of size $k$, one can construct in polynomial time a feasible solution of size $k + 2n + 1$ for $\mathcal{I}'$. For every $i\in U$, let $p_i$ be the smallest number such that $S_{p_i}\in F'$ and $i\in S_{p_i}$. We can construct a solution for $\mathcal{I'}$ by taking all the arcs $(v_{p_i},u_i)$, all the edges that are incident to $t$, and all the edges $(s,v_j)$ such that $S_j\in F'$. Note that 
    the size of this solution is $n+(n+1)+|F'|$ and it is feasible as any terminal can reach $s$ and also can be reached from $s$. 
    
    We also show that given a feasible solution $A'\subseteq A(D)$ for \(\mathcal{I'}\) of size $k$, we can find a solution of size at most \(k-2n-1\) for \textsc{Set Cover}. Note that as $A'$ is feasible, it must contain all the edges incident on $t$ as it is the only out-neighbor of all the terminals. Moreover, in $A'$, $s$ must have a path to all the terminals and any path from $s$ to a terminal $u_i$ in $D$ is a path of length two whose middle point in $V_F$. Let $V' \subseteq V_F$ be the set of all middle points of these paths. Clearly, the number of edges that appear in all these paths are at least $|V'|+n$, as it contains all the arcs from $t$ to any vertex in $V'$ and also at least one arc going to any terminal $u_i$. Now the family of sets $F'\subseteq F$ that correspond to $V'$ covers $U$, which means it is a solution of size at most $k-2n-1$ for $\mathcal I$.
\end{proof}

\section{Future Research Directions}\label{sec:con}
In this paper we gave new algorithms for SCSS and related problems. We believe there are ample directions for further research and outline some of them now.

First of all, it is natural to investigate how well our method works for problems involving $p$-edge-connectivity, which can be approached with a similar characterization via in/out-branchings~\cite{DBLP:journals/algorithmica/AgrawalMPS22}. For example, how fast can the \textsc{$p$-ECSS} problem be solved parameterized by the treewidth (see~\cite{DBLP:journals/algorithmica/AgrawalMPS22} for a problem definition)?

Another natural question is to get "Strong ETH-tight" algorithms for \textsc{SCSS}. Something that has been achieved for many natural problems (see e.g.~\cite{LokshtanovMS18}).
\begin{question}\label{q:tight}
    Find a $O^*(c^{\tw})$ time algorithm for SCSS and a proof that a $O^*((c-\eps)^{\tw})$ time algorithm for any $\varepsilon >0$ refutes the Strong Exponential Time Hypothesis.    
\end{question}

A probably easier version of Question~\ref{q:tight} is obtained by replacing $\tw$ with the \emph{pathwidth} of the graph. To attack this question, we believe the following combinatorial question is central:

\begin{question}\label{q:com}
    Let $(V,A)$ be a complete directed graph and let $A_1,\ldots,A_\ell \subseteq A$ and $B_1,\ldots,B_\ell \subseteq A$ such that  $(V , A_i \cup B_j)$ is strongly-connected if and only if $i=j$. How large can $\ell$ be in terms of $|V|$?
\end{question}
The quantity $\ell$ above coincides with the \emph{largest induced matching} of the \emph{compatibility matrix} associated with the SCSS problem (see e.g.~\cite{Nederlof20} for definitions), and it is natural to conjecture a close connection of it with the complexity of SCSS parameterized by pathwidth:\begin{conjecture}
If $\ell = \theta(c^{|V|})$ in Question~\ref{q:com}, then \textsc{SCSS} can be solved in $(c+1)^{\pw}$ time and not in $(c+1-\epsilon)^{\pw}$, assuming SETH.
\end{conjecture}

Note that, in the above conjecture the $+1$ is because vertices can also not participate in the solution.

In the $2$VCSS problem one is given an undirected graph and seeks a subgraph with minimum number of edges that is $2$-vertex connected. We leave the following problem open:
\begin{question}\label{q:2vcss}
    Can the $2$-VCSS problem be solved in $O^*(c^{O(\tw)})$ time?
\end{question}

Finally, it is natural to further study the exact exponential time complexity of the problems studied in this paper.

\begin{question}\label{q:2ecssSubQuadratic}
    Can \textsc{$2$-ECSS} be solved in $O^*((2-\eps)^n)$ time, for some $\eps>0$?
\end{question}

Since $\textsc{$2$-ECSS}$ generalizes the \textsc{Hamiltonicity} problem on undirected graphs, such an algorithm probably needs randomization and algebraic techniques such as the ones from~\cite{Bjorklund14}.

\bibliographystyle{abbrv}
\bibliography{ref}

\appendix

\section{Reducing $2$-\textsc{ECSS} to \textsc{SCSpS}}\label{sec:red}
\begin{lemma}\label{lem:ecssred}
    An instance $(G=(V,E),t)$ of the $2$-\textsc{ECSS} problem can be reduced to an equivalent instance of the \textsc{SCSpS} problem $(D=(V,A),t)$ by adding for each edge $\{u,v\} \in E$ with the two arcs $(u,v)$ and $(v,u)$ to $A$.
\end{lemma}
\begin{proof}
    If $G$ is not $2$-edge connected, then the instance is trivially a \no-instance, so assume that $G$ is $2$-edge connected. The forward direction follows from Robbins' theorem: every $2$-edge connected spanning subgraph admits a strongly connected orientation.
    For the converse direction, consider a minimum strongly connected spanning subdigraph of $D$ whose underlying undirected graph has the minimum number of bridges. If the underlying undirected graph has a bridge $e$, then both orientations of $e$ must be used. Since the original graph is $2$-edge connected, there is another edge crossing the same cut. Replacing one orientation of $e$ by a suitable orientation of that edge preserves strong connectivity and decreases the number of bridges, a contradiction. Hence the graph is bridgeless, that is, $2$-edge connected.
\end{proof}

\end{document}